\documentclass[11pt]{article}

\usepackage[a4paper, total={6in, 8in}]{geometry}

\usepackage{amsfonts}
\usepackage{graphicx}
\usepackage{epstopdf}
\usepackage{algorithm}
\usepackage{algpseudocode} 
\usepackage{xcolor}
\usepackage{stfloats}
\usepackage{multirow}
\usepackage{amsopn}
\usepackage{graphicx} 
\usepackage{amsfonts}
\usepackage{amsmath}   
\usepackage{amssymb}
\usepackage{mathtools}
\usepackage{enumerate}
\usepackage{hyperref}
\usepackage{caption}  
\usepackage{subcaption}
\usepackage{fullpage}  

\usepackage{tikz}
\usetikzlibrary{arrows,positioning,shapes,3d,calc}
\usepackage{pgfplots}
\usepgfplotslibrary{groupplots}

\pgfplotsset{compat = newest,
		plot coordinates/math parser = false}
\newlength\figureheight
\newlength\figurewidth

\usepackage{amsthm}
\newtheorem{theorem}{Theorem}[section]

\newtheorem{lemma}[theorem]{Lemma}

\newtheorem{definition}[theorem]{Definition}

\def\ones{{\bf 1}}
\def\diag{{\bf diag}}

\DeclareMathOperator*{\argmin}{argmin}
\DeclareMathOperator*{\argmax}{argmax}
\DeclareMathOperator{\tr}{trace}

\providecommand{\keywords}[1]
{
  \small	
  \textbf{\textit{Keywords---}} #1
}

\title{Community Detection by a Riemannian Projected Proximal Gradient Method}

\author{Meng Wei\thanks{Department of Mathematics, Florida State University, 208 Love Building, 1017 Academic Way, Tallahassee, FL 32306-4510, USA.
            \href{mailto:mwei@math.fsu.edu}{\texttt{mwei@math.fsu.edu}}.}
\and Wen Huang\thanks{Corresponding author. School of Mathematical Sciences, Fujian Provincial Key Laboratory of Mathematical Modeling and High-Performance Scientific Computing,
        Xiamen University, Xiamen, Fujian, P.R.China, 361005. 
        \href{mailto:wen.huang@xmu.edu.cn}{\texttt{wen.huang@xmu.edu.cn}}.}
\and Kyle A. Gallivan\thanks{Department of Mathematics, Florida State University, 208 Love Building, 1017 Academic Way, Tallahassee, FL 32306-4510, USA.
	\href{mailto:gallivan@math.fsu.edu}{\texttt{gallivan@math.fsu.edu}}.}
\and Paul Van Dooren\thanks{Department of Mathematical Engineering, Universit\'e catholique de Louvain, Louvain-La-Neuve, Belgium.
	\href{mailto:paul.vandooren@uclouvain.be}{\texttt{paul.vandooren@uclouvain.be}}.}
}

\begin{document}

\maketitle

\begin{abstract}
Community detection plays an important role in understanding and exploiting the structure of complex systems. Many algorithms have been developed  for community detection using modularity maximization or other techniques. In this paper, we formulate the community detection problem as a constrained nonsmooth optimization problem on the compact Stiefel manifold. A Riemannian projected proximal gradient method is proposed and used to solve the problem. To the best of our knowledge, this is the first attempt to use Riemannian optimization for community detection problem. Numerical experimental results on synthetic benchmarks and real-world networks show that our algorithm is effective and outperforms several state-of-art algorithms.
\end{abstract}




\section{Introduction}
Describing and analyzing complex systems in mathematical models is a challenging problem. 
Networks are a natural representation for many kinds of complex systems, where networks are sets of nodes or vertices joined together in pairs by links or edges. There are several types of networks.  For example, Facebook is a large social network, where more than one billion people are connected via virtual acquaintanceship. Another common example is the internet, the physical network of computers, routers, and modems which are linked via cables or wireless signals. Many other examples come from biology, physics, engineering, computer science, ecology, economics, marketing, etc. 

Real-world networked systems often have a community structure, which is the division of network nodes into groups such that the network connections are denser within the groups and are sparser between the groups, see \cite{newman2004finding}. These groups are called communities, or modules. 

Detecting community structure in a network is a powerful tool for understanding and exploiting the structure of networks, and it has various practical applications \cite{girvan2002community}. Communities in a social network might represent real social groupings, perhaps by acquaintanceship, interest or background; communities in a metabolic network might represent cycles and other functional groupings; communities on the web might represent pages on related topics. 

A variety of community detection algorithms have been developed in recent years, such as the GN algorithm \cite{newman2004fast},  the spectral modularity maximization algorithm \cite{newman2006modularity}, the Louvain method \cite{blondel2008fast}, the Infomap algorithm \cite{rosvall2008maps}, statistical inference \cite{newman2007mixture}, deep learning \cite{yang2016modularity}. Modularity optimization approaches have been shown to be highly effective in practical applications. \cite{fortunato2010community} covers in practical and theoretical detail modularity-based approaches to  community detection.

Recently, optimization over Riemannian manifolds has drawn much attention because of its application in many different fields. Almost all of the manifold optimization methods require computing the derivatives of the objective function and do not apply to the case where the objective function is nonsmooth. In \cite{chen2020proximal}, the authors proposed a Riemannian proximal gradient method called ManPG for a class of nonsmooth nonconvex optimization problems over a Stiefel manifold
\begin{gather}
\min F(X):=f(X)+g(X), \\
s.t. ~X\in \mathcal{M}:=St(q,n)=\{X:X\in \mathbb{R}^{n\times q}, X^TX=I_q\},\notag
\label{nonsmooth}
\end{gather}
where $I_q$ denotes the $q\times q$ identity matrix ($q<n$), $f$ is smooth, possibly nonconvex, and its gradient $\nabla f$ is Lipschitz continuous, $g$ is convex, possibly nonsmooth, and is Lipschitz continuous and the proximal mapping of $g$ is easy to find.

In \cite{huang2019extending}, the authors extended the fast iterative shrinkage-thresholding (FISTA) algorithm to solve (\ref{nonsmooth}), and the accelerated Riemannian manifold proximal gradient algorithm performed better than ManPG. In \cite{huang2019riemannian}, they developed and analyzed a generalization of the proximal gradient methods with and without acceleration for nonsmooth Riemannian optimization problems.

In this paper, we propose the accelerated Riemannian manifold projected proximal gradient (ARPPG) method for community detection, and we solve the community detection problem using a constrained nonsmooth optimization problem over a Stiefel manifold.

The paper is organized as follows. In Section 2, we define assignment matrices and show that the ideal graph assignment is a global maximal solution of the modularity function. In Section 3, we show the connection between the modularity matrix and the Stiefel manifold, and then transform the community detection problem to the constrained Stiefel optimization problem. Because the constraint defines a feasible set that is a subset of the Stiefel manifold, we must apply a projection to the proximal result.  This leads to the accelerated Riemannian manifold projected proximal gradient (ARPPG) algorithm.  Extensive numerical experiments on synthetic and real world networks are described in Section 4. Finally, conclusions and future work are stated in Section 5.

\section{Derivation of Global Maximum over Assignment Matrices}

\subsection{Assignment matrices}
We will denote a $q$ dimensional vector with all entries being $1$ by $\ones_q$ and denote the $q\times q$ permutation matrices by $P_q$.

A matrix in the set of assignment matrices, ${\cal A}_{n,q}$, is defined as
\begin{definition}
The matrix $X \in \lbrace 0, 1 \rbrace^{n \times q}$, with $n \geq q$, is an assignment matrix if it satisfies
\begin{enumerate}[(i)]
\item $X \ones_q = \ones_n$,
\item $X^TX = \diag(n_1, \cdots , n_q)$
where $n_i=||Xe_i||_1$.
\end{enumerate}

$X$ is said to be in canonical ordering if the rows
are permuted so that 
\[
X = 
\begin{pmatrix} 
\ones_{n_1} & & & \\
& \ones_{n_2} & & \\
& & \ddots & \\
& & & \ones_{n_q} 
\end{pmatrix}.
\]
\end{definition}
Of course, the column ordering is not unique for the canonical form, i.e.,
$XP_q$ is the same community assignment but with a different
correspondence between the sets and the columns of the assignment matrix.
For essential uniqueness, the additional constraint of
$n_1 \geq n_2 \geq \dotso \geq n_q$ can be imposed.
The columns are orthogonal, but not orthonormal, and $X$ has exactly
$n$ nonzero elements all of which have the value of $1$. As a result,
$X$ defines a partitioning of the indices $1, \dotsc , n$ into $q$
disjoint sets.

\subsection{The Modularity Cost Function}
From \cite{newman2006finding}, the scalar cost function $f(X)$ called modularity (up to a scalar $\frac{1}{2m}$) can be written as a quadratic function over
$n \times q$ matrices defined by the matrix
\[
M = A - \frac{A \ones_n \ones_n^T A}{2m},\;\;f(X)=\tr(X^TMX),
\]
where $A$ is the adjacency matrix of the graph, $M$ is the modularity matrix, $m$ is the number of edges, $n$ is the number of vertices in the graph and the total degree of the graph is $2m=\ones_n^T A  \ones_n$.

The value of $f(X)$ is invariant under permutations on the columns of 
the assignment matrix $X$, i.e., $f(XP_q)=\tr(P_q^TX^TMXP_q)$. So there are
multiple optimal ways of specifying the same community assignment.

\subsection{Maximal of the Modularity Function on Ideal Graphs}
We consider in this section so-called ideal graphs. An ideal graph is a graph where the communities are cliques and there are no edges between the cliques.

When $A$ is an ideal graph with $q$ communities we know it can be written \cite{marchand2017low}
\[
A = \tilde{Z}_* \tilde{Z}_*^T
\]
where  $\tilde{Z}_* \in {\cal A}_{n,q}$ is not necessarily in canonical
form and there exists a row permutation $P$ so that
\[
PAP^T = A_P=Z_* Z_*^T
\]
where $A_P$ is block diagonal with diagonal blocks $\ones_{n_i}\ones_{n_i}^T
= z_iz_i^T$ for $1 \leq i \leq q$ and
\[
Z_* = 
\begin{pmatrix} 
\ones_{n_1} & & & \\
&\ones_{n_2} & & \\
& & \ddots & \\
& & &\ones_{n_q} 
\end{pmatrix}
= \begin{pmatrix} 
z_1& & & \\
&z_2 & & \\
& & \ddots & \\
& & &z_q 
\end{pmatrix}
\]
is in canonical form.

The corresponding modularity matrices for an ideal $A$ and the corresponding block diagonal $A_P$ are given by
\begin{align*}
M& = M^T= A - \frac{A \ones_n \ones_n^T A}{2m}  \\
&= \tilde{Z}_* \tilde{Z}_*^T - \frac{\tilde{Z}_* \tilde{Z}_*^T  
\ones_n \ones_n^T \tilde{Z}_* \tilde{Z}_*^T}{2m}   \\
&= \tilde{Z}_*(I_q - \frac{\tilde{s}\tilde{s}^T}{2m} )\tilde{Z}_*^T, \\
\end{align*}
and 
\begin{equation*}
M_P= {Z}_*(I_q -\frac{ss^T}{2m}){Z}_*^T,
\end{equation*}
where $s = Z_*^T\ones_n =\begin{pmatrix}n_1 &  \dotsc & n_q \end{pmatrix}^T$, and $2m = \ones_n^T A_P  \ones_n = s^Ts = \sum_{i=1}^q n_i^2$.

The cost function $f(X)$ is invariant under reorderings of $A$, so
we can analyze any row ordering of $Z_*$ denoted generically as $Z$
below.
The following result for the value of $f(Z)$, i.e., the
cost function at the assignment matrix that generates the ideal matrix $A$,
follows directly from the definitions.

\begin{lemma}
If $A = ZZ^T$ for $Z \in {\cal A}_{n,q}$ then
\[
f(Z) = \sum_{i=1}^q n_i^2 - \frac{\sum_{i=1}^q n_i^4}{\sum_{i=1}^q n_i^2}.
\]
\end{lemma}

We show that the value $f(X)$ for any $X \in {\cal A}_{n,q}$ is bounded above
by $f(Z)$ in Theorem \ref{thm max}. The following lemmas are easily proven and are useful in proving the main result.

\begin{lemma}
If $A = ZZ^T$ for $Z \in {\cal A}_{n,q}$ then for any 
$X \in {\cal A}_{n,q}$ 
\[
f(X)= \tr(X^TZ(I_q - \frac{ss^T}{2m}) Z^TX) \leq \sum_{i=1}^q \gamma_i v_i^TXX^Tv_i,
\]
where $s = \begin{pmatrix} n_1 & \dotsc & n_q \end{pmatrix}^T$,
$v_i = Ze_i$,
\[
2m =\sum_{i=1}^q n_i^2, ~\gamma_i := 1 - \frac{n_i^2}{2m},
\]
where $0 \leq \gamma_i < 1$.
\end{lemma}

%

\begin{lemma}\label{lem:zmax}
Given $Z \in {\cal A}_{n,q}$, any 
$X \in {\cal A}_{n,q}$ satisfies
\[
v_i^TXX^Tv_i \leq v_i^TZZ^Tv_i,\;\;\;1 \leq i \leq q
\]
where $v_i = Ze_i$.  Equality holds only when $X=ZP_q$, i.e., a column permutation
of $Z$.
\end{lemma}

%

The desired result is stated as Theorem~\ref{thm max}.
\begin{theorem}\label{thm max}
If $A = ZZ^T$ for $Z \in {\cal A}_{n,q}$ is an ideal adjacency matrix then
for any $X \in {\cal A}_{n,q}$ 
\[
f(X) \leq f(Z),
\]
where $f(X) = \tr(X^TZ(I_q - \frac{ss^T}{2m}) Z^TX)$,
$s = \begin{pmatrix} n_1 & \dotsc & n_q \end{pmatrix}^T$, 
$2m={\sum_{i=1}^q n_i^2}$.
\end{theorem}
\begin{proof}
The series of lemmas above yields
\[
f(X) \leq \sum_{i=1}^q \left(1 - \frac{n_i^2}{2m}\right) v_i^TZZ^Tv_i.
\]
Note that 
\begin{align*}
&\sum_{i=1}^q \left(1 - \frac{n_i^2}{2m}\right) v_i^TZZ^Tv_i \\
&=\sum_{i=1}^q v_i^TZZ^Tv_i - \sum_{i=1}^q\frac{n_i^2}{2m} v_i^TZZ^Tv_i \\
&= \sum_{i=1}^q n_i^2- \frac{\sum_{i=1}^q n_i^4}{2m}\\
&= \sum_{i=1}^q n_i^2- \frac{\sum_{i=1}^q n_i^4}{\sum_{i=1}^q n_i^2}
=f(Z).
\end{align*}
\end{proof}
Theorem \ref{thm max} shows that the ideal graph assignment is a global maximum of the modularity function over $\mathcal{A}_{n,q}$.

\section{Stiefel Manifold Algorithms for Community Detection}

The algorithms discussed here assume the cost function
\[
f(X)=\tr(X^TMX),\]
\[
where~~ M = A - \frac{A\ones_n\ones_n^TA }{\ones_n^TA\ones_n}.
\]

\subsection{The connection between the modularity matrix and the Stiefel manifold $St(q,n)$}

\begin{lemma}\label{rangeA_M_1}
 Let $Z\in {\mathcal A}_{n,q}$ and define $M=A-\frac{A\mathbf{1}\mathbf{1}^\mathsf{T}A}{\mathbf{1}^\mathsf{T}A\mathbf{1}}$.
 If $A$ is the adjacency matrix of an ideal graph, then $A=ZZ^T$ and
 \begin{equation} \label{1}
 \mathcal R(A) =  \mathcal R(Z) = (\mathcal R(M) \oplus^\perp \mathcal R(\mathbf{1}_n)),
 \end{equation}
 \begin{equation} \label{2} 
 \mathcal N(M) =  (\mathcal N(Z^T) \oplus^\perp \mathcal R(\mathbf{1}_n)) =  (\mathcal N(A) \oplus^\perp \mathcal R(\mathbf{1}_n)),
 \end{equation}
 where $ \mathcal R(A) $ denotes the range of $A$, $\mathcal N(A) $ denotes the null space of $A$, and $\oplus^\perp$ denotes the direct sum of two perpendicular spaces.
 \end{lemma}
 \begin{proof}
  Note that the symmetry of $A$ and $M$ implies $\mathcal N(M)=\mathcal R(M)^\perp$ and $\mathcal N(A)=\mathcal R(A)^\perp$. Therefore 
  \eqref{1} and \eqref{2} are equivalent. It follows from the definition of $M$ that 
\begin{align*}
 M&= Z(I_q-\frac{ss^T}{s^Ts}) Z^T=Z(I_q-\frac{ss^T}{s^Ts})^2 Z^T\\
  &=[Z(I_q-\frac{ss^T}{s^Ts})][(I_q-\frac{ss^T}{s^Ts})Z^T].
\end{align*}  
 
  This implies that
  $$ \mathcal R(M)=\mathcal R(Z(I_q-\frac{ss^T}{s^Ts})), \quad \mathcal N(M)=\mathcal N((I_q-\frac{ss^T}{s^Ts})Z^T),
  $$
  and we also have 
  $$ \mathcal R(A)=\mathcal R(Z), \quad \mathcal N(A)=\mathcal N(Z^T).
  $$
  Since the projector $(I_q-\frac{ss^T}{s^Ts})$ has rank $n-1$, it follows that the ranges and null spaces of $A$ and $M$ have dimensions that can only differ by 1 at most.
  Now consider the vector $\mathbf{1}_n$. Since $s=Z^T\mathbf{1}_n$, we have $M\mathbf{1}_n=0$ and since $M$ is symmetric, $\mathbf{1}_n$ is orthogonal to $\mathcal R(M)$.
  Since $\mathbf{1}_n=Z\mathbf{1}_q$, we have $\mathbf{1}_n\in\mathcal R(Z)=\mathcal R(A)$. Together, these two properties prove \eqref{1}, and hence also \eqref{2}.
 
 \end{proof}
  
Since  $rank(M)=q-1$ and $M=M^T$, we have the eigendecomposition
\[
M = X_* \Gamma X_*^T,\]
where
\[X_* \in St(q-1,n),\;\;\Gamma = diag(\gamma_1, \dotsc,
\gamma_{q-1}), \;\;\gamma_i \neq 0 .
\]
It then follows by Lemma \ref{rangeA_M_1} that $\left[ X_* ~ \frac{\mathbf{1}_n}{\sqrt{n}} \right] \in St(q,n)$ since $\mathcal R(X_*)$ is a subspace of $\mathcal R(M)$.

\subsection{An Important Basis for an Ideal ${\cal R}(A)$}

For the modularity matrix, the relationship between $A$ and $M$ is one of deflation of range that allows
the characterization of the part of ${\cal R}(A)= {\cal R}(Z)$ that is removed
when considering ${\cal R}(M)$ as shown in \eqref{1}. 

Therefore, we can now get the anticipated result of  ${\cal R}(Z) = 
{\cal R}\left(\begin{bmatrix} X_* & \frac{\ones_n}{\sqrt{n}} \end{bmatrix}\right)$.
\subsection{A Constrained Stiefel Optimization Problem}
\subsubsection{Multiple Extrema:}
Note that if a space $\cal B$ of dimension $q$ has a basis that is an assignment matrix then
it has $q!$ such bases all of which are of the form $ZP_q$ where
$Z$ is any assignment matrix basis and $P \in \{0,1\}^{q \times q}$ is a permutation
matrix.  All of these matrices have exactly $n$ nonzero elements which is the minimum count possible for bases of the space. If the columns of such a matrix, $Z$, are normalized in Euclidean 2-norm length then an element of $St(q,n)$ is produced with 
$n_i$ elements in column $i$ all with the value $1/\sqrt{n_i}$ with $\sum_{i=1}^q
n_i = n$. These are the global minima of
\[
\min_{X \in St(q,n), {\cal R}(X)={\cal B}} \lVert X \rVert_1,
\]
where the $l_1$ norm is defined as $\lVert X \rVert_1=\sum_{ij}\lVert X_{ij} \rVert$ imposing the sparsity of $X$.

In practical numerical computation, even on ideal matrices and certainly on problems
for which noise perturbs $A$ and $Z$ from ideal, some projection is needed to take a matrix
in $St(q,n)$ to the ``nearest'' matrix in ${\cal A}_{n,q}$. 

\subsubsection{A Constrained Stiefel Optimization Problem:}

The constrained Stiefel optimization problem used to perform community detection is
\begin{equation}
X_*=\argmax_{X \in St(q,n),\;\ones_n \in {\cal R}(X)} \tr(X^TMX) - \lambda 
\lVert X \rVert_1,
\label{constrainedCostFunction}
\end{equation}
where $\lambda>0$ is a tuning parameter controlling the balance between variance and sparsity. The approach to compute $X_*$ is given in Algorithm \ref{continuation1}.

\begin{algorithm}
	\caption{Algorithm for the Constrained Stiefel Optimization Problem} 
	\begin{algorithmic}[1]
		\State \textit{Step 1:} Compute  $Y_* \in St(q-1,n)$ where
\[
Y_* = \argmax_{X \in St(q-1,n)} \tr(X^TMX).
\] and
set the initial guess for Step 2 as
\[
X_0 = \begin{bmatrix} Y_* & \frac{\ones_n}{\sqrt{n}} \end{bmatrix} .
\]
\State \textit{Step 2:}
Compute $X_* \in St(q,n), \ones_n \in {\cal R}(X)$  where
\begin{equation*}
X_*=\argmax_{X \in St(q,n),\;\ones_n \in {\cal R}(X)} \tr(X^TMX) - \lambda 
\lVert X \rVert_1,
\end{equation*}
with $X_0$ as the initial guess.

\State \textit{Step 3:}
Get the assignment matrix $\hat{X}_*$ by setting the element with the largest magnitude in each row of $X_*$ as $1$, and the others as $0$ when $X_*$ is sufficiently sparse.
Assess the assignment matrix $\hat{X}_*$ and determine whether it is acceptable as a solution
to the community detection problem or if the parameter $\lambda$ should be updated.
If $\lambda$ is updated then return to Step 2.
   
	\end{algorithmic} 
	\label{continuation1}
\end{algorithm}

Step 1 can be computed using any trace maximization algorithm. Our code uses
RNewton in ROPTLIB \cite{huang2018roptlib}.
A projection is needed to define a Riemannian projected proximal gradient algorithm to solve this problem in Step 2. 
In fact, this projection can be used for any line search based algorithm where
$Y_k = R(\alpha D_k)$  for a Riemannian retraction $R$ must be feasible.
The projection used in the proposed algorithm is described below.

In Step 3, we use the idea of continuation to choose the parameter $\lambda$ that defines the cost function. We can get the optimal $X^*_1$ after setting the initial $\lambda_0$ and $X_0$. We then increase $\lambda_0$ and  use $X^*_1$ as the initial matrix to get $X^*_2$. We continue this procedure until the cost function $\tr(X^TMX) - \lambda \lVert X \rVert_1$ does not improve anymore.

Step 2 is the main part of the algorithm, and it is inspired by \cite{huang2019extending}. 
In \cite{huang2019extending}, the authors generalized the FISTA \cite{beck2009fast} from the Euclidean space to the Riemannian setting and considered the general nonconvex optimization problem  
\begin{equation}
\min_{X\in \mathcal{M}} F(X) = f(X) + g(X),
\label{eq: general}
\end{equation}
where $\mathcal{M} \subset \mathbb{R}^{n\times q}$ is a Riemannian submanifold, $f:\mathbb{R}^{n\times q} \rightarrow \mathbb{R}$ is $L$-continuously differentiable (may be nonconvex) and $g$ is continuous and convex but may not be differentiable.

The optimization problem \eqref{constrainedCostFunction} is a special case of the problem \eqref{eq: general}, where $f(X)=\tr(X^TMX)$ is $L$-continuously differentiable and $g(X)=- \lambda 
\lVert X \rVert_1$ is continuous, convex, but not differentiable. However, there is an essential difference between \eqref{eq: general} and \eqref{constrainedCostFunction} in that there is a constraint $\ones_n \in {\cal R}(X)$ that defines a feasible set ${\cal F} \subset St(q,n)$. The accelerated Riemannian manifold proximal gradient method \cite{huang2019extending} is modified to define the accelerated Riemannian manifold projected proximal gradient (ARPPG) method by adding the projection (\ref{projection}) derived in the next section. The details of ARPPG are in Algorithm \ref{alg:ARPPG}.

\begin{algorithm}
	\caption{Accelerated Riemannian Manifold Projected Proximal Gradient Method(ARPPG)} 

	\hspace*{\algorithmicindent} \textbf{Input:} Lipschitz constant $L$ on $\nabla f$, parameter $\mu \in (0, 1/L]$ in the proximal mapping, line search parameter $\sigma \in (0,1)$, shrinking parameter in line search $\beta\in (0,1)$, positive integer $N$ for safeguard;\\
	\begin{algorithmic}[1]
	\State $t_0 = 1, y_0 = x_0, z_0 = x_0; \lambda=\lambda_0$
		\For {$k=0,...$} 
		    \If {mod(k, N) = 0} \qquad $\triangleright$ Invoke safeguard every $N$ iterations
			    \State Invoke Algorithm 3: $[z_{k+N}, x_k, y_k, t_k]=Algo3(z_k, x_k, y_k, t_k, F(x_k))$;
			\EndIf
			\State Compute $$\eta_{y_k} = \argmin_{\eta \in T_{y_k}\cal{M}}\langle \text{grad} f(y_k), \eta\rangle + \frac{1}{2\mu}||\eta||_F^2 + g(y_k + \eta);$$
			\State $x_{k+1} = R_{y_k}(\eta_{y_k})$;
			\State $x_{k+1} = proj(x_{k+1})$;
			\State $t_{k+1} = \frac{\sqrt{4t_k^2+1}+1}{2}$;
			\State Compute $$ y_{k+1} = R_{x_{k+1}}(\frac{1-t_k}{t_{k+1}}R_{x_{k+1}}^{-1}(x_k));$$
			\State Compute $y_{k+1} = proj(y_{k+1})$.
		\EndFor
		\State $X_* = x_{k+1}$
	\end{algorithmic} 
    \label{alg:ARPPG}
\end{algorithm}

\begin{algorithm}
	\caption{Safeguard for Algorithm ARPPG} 
	\hspace*{\algorithmicindent} \textbf{Input:} [$z_k, x_k, y_k, t_k, F(x_k)$];\\
	\hspace*{\algorithmicindent} \textbf{Output:} [$z_{k+N}, x_k, y_k, t_k$];\\
	\begin{algorithmic}[1]
	\State Compute $$\eta_{z_k}=\argmin_{\eta\in T_{z_k}\cal{M}}\langle \text{grad} f(z_k), \eta\rangle + \frac{1}{2\mu}||\eta||_F^2 + g(z_k+\eta);$$
	\State Set $\alpha=1$;
	\While {$F(proj(R_{z_k}(\alpha \eta_{z_k})))>F(z_k)-\sigma\alpha||\eta_{z_k}||_F^2$}
	    \State $\alpha = \beta\alpha$;
	\EndWhile
	\If {$F(proj(R_{z_k}(\alpha \eta_{z_k})))<F(x_k)$} \qquad $\triangleright$ Safeguard takes effect
	    \State $x_k=R_{z_k}(\alpha\eta_{z_k}), ~ y_k=R_{z_k}(\alpha\eta_{z_k})$, and $t_k=1$;
	    \State $x_k=proj(x_k), ~y_k=proj(y_k)$;
	\Else
	    \State $x_k, y_k$ and $t_k$ keep unchanged;
	\EndIf
	\State $z_{k+N}=x_k$; \qquad $\triangleright$ Update the compared iterate
	\end{algorithmic} 
	\label{safeguard}

\end{algorithm}


There are several retractions that can be constructed for the Stiefel manifold. Algorithm \ref{alg:ARPPG}, uses the efficient retraction in \cite{huang2019extending} based on the singular value decomposition (SVD):
\begin{gather*}
[Q, R]= \text{qr}(X+\eta_X),~ [U, S, V]=\text{svd}(R), \\
R_X(\eta_X)=Q(UV^T),
\end{gather*}

where qr and svd mean computing the compact QR decomposition and SVD of a matrix, respectively. $R_{X}^{-1}(Y)=YS-X$, where $S$ is the solution of the Lyapunov equation $(X^TY)S+S(Y^TX)=2I_q$.

\subsubsection{The Projection:}

Given $X \in St(q,n)$, the task is to find 
a $Y \in St(q,n)$ with $\ones_n \in {\cal R}(Y)$ that minimizes
$\lVert X - Y \rVert_F^2$. 
Letting $f(Y;\;X) = \tr (X^TY)$ denote a cost function
parameterized by $X$, the problem can be formulated in an equivalent form
by noting
\[
\min_{Y \in {\cal F}} \lVert X - Y \rVert_F^2 \leftrightarrow 
\max_{Y \in {\cal F}} f(Y;\;X) 
\]
where ${\cal F} = \lbrace Y \in St(q,n),\;\;\ones_n \in {\cal R}(Y) \rbrace$.
The maximum value of $f(Y;\;X)$ is $q$ and is achieved when $X \in {\cal F}$ and the problem is invariant with respect to $Q \in {\cal O}(q)$ where ${\cal O}(q)$ is the orthogonal group consisting of $q$-by-$q$ orthogonal matrices, i.e.,
\[
f(Y;\;X) = f(YQ;\;XQ).
\]
Note that the cost function changes for this invariance. In general,
$f(Y;\;X) \neq f(YQ;\;X)$.

For an element of the feasible set ${\cal F}$, there must exist
$Q \in {\cal O}(q)$ such that $\hat{Y} = YQ =  \begin{bmatrix} \tilde{1}_n & \tilde{Y} \end{bmatrix}$ with $\tilde{1}_n = \ones_n / \sqrt{n}$, $\tilde{Y} \in St(q-1,n)$ and
${\cal R}(\tilde{Y}) \perp \tilde{1}_n$. There are, of course, many such $\hat{Y}$ possible.
This can be seen from
\begin{gather*}
Q= \begin{bmatrix} q_1 & Q_\perp \end{bmatrix},\;\;
\hat{Y} = \begin{bmatrix} \tilde{1}_n & \tilde{Y} \end{bmatrix}
= \begin{bmatrix} Yq_1 & {Y}Q_\perp \end{bmatrix}.
\end{gather*}
Given $Y$, the vector $q_1$ is uniquely defined but $Q_\perp$ is 
any orthonormal completion of $q_1$ and $\tilde{Y} = YQ_\perp$ varies with 
the choice of $Q_\perp$.
This can be used to parameterize the cost function over ${\cal F}$ to
give an alternative form of the optimization problem defining the projection
and reveal a constructive form of the solution $Y_*$.

The two forms of the optimization problem are 
\begin{gather*}
Y_* = \argmax_{Y \in {\cal F}} \tr (X^TY),\\
where~ {\cal F} = \lbrace Y \in St(q,n),\;\;\ones_n \in {\cal R}(Y) \rbrace;\\
\\
\left(\hat{Y}_*, Q_*\right)= \argmax_{\hat{Y} \in {\cal G},\; Q \in {\cal O}(q)} \tr (Q^TX^T\hat{Y}),\\
{\cal G} = \lbrace \hat{Y}=\begin{bmatrix} \tilde{\ones}_n & \tilde{Y} \end{bmatrix} 
\;\vert \; \tilde{Y} \in St(q-1,n),\;\;\ones_n \perp {\cal R}(\tilde{Y}) \rbrace.
\end{gather*}

The second form can be solved analytically and a solution for the first form recovered
easily.
The cost function for the second form can be expanded as
\[
\tr(Q^T{X}^T\hat{Y}) = q_1^TX^T \tilde{\ones}_n + \tr (Q_\perp^TX^T\tilde{Y}).
\]
The first term of the sum in the cost function is independent of the second term while the second
term is essentially determined by the choice of $q_1$. For any $q_1$
and orthonormal completion $Q_\perp$,
the maximum value of $q_1$  for the second term 
is achieved by $\tilde{Y} = XQ_\perp$. 

Given this optimal choice of $\tilde{Y}$ parameterized by $Q$, the problem then becomes
finding the optimal $q_*$ for
\[
\max_{q_1 \in St(1,n)} \;\;q_1^T X^T\tilde{\ones}_n  
\]
and $Q_*=\begin{bmatrix} q_* & Q_\perp^* \end{bmatrix}$ where $Q_\perp^*$ is any
orthonormal completion of $q_*$.
This has a maximum value of $1$ if and only if $\tilde{\ones}_n \in {\cal R}(X)$.
Otherwise it is  maximized by
\[
q_* = \frac{X^T\tilde{\ones}_n}{\lVert X^T\tilde{\ones}_n \rVert_2}.
\]
There are several maximizers given by
\begin{gather*}
q_* = \frac{X^T\tilde{\ones}_n}{\lVert X^T\tilde{\ones}_n \rVert_2} \\
Q_\perp^* \in St(q-1,n) \;\;\text{is any orthonormal completion of $q_*$} \\
\tilde{Y}_*   = XQ_\perp^*,~ \hat{Y}_* =\begin{bmatrix} \tilde{\ones}_n & \tilde{Y}_* \end{bmatrix}.\\
\end{gather*}

Finally, $Y_*$, the maximizer for the original parameterized form
of $f(Y;\;X)$ can be determined from $\hat{Y}_*$
\begin{gather*}
Y_* = \hat{Y}_* Q_*^T = \begin{bmatrix} \tilde{\ones}_n & \tilde{Y}_* \end{bmatrix}
\begin{bmatrix} q_* & Q_\perp^* \end{bmatrix}^T
= \tilde{\ones}_n q_*^T + X Q_\perp^* (Q_\perp^*)^T.
\end{gather*}
This form shows that the choice of $Q_\perp^*$, i.e.,
the basis for ${\cal R}^\perp(q_*)$, that determines $\hat{Y}_*$
does not result in multiple $Y_*$
since the projector $Q_\perp^* (Q_\perp^*)^T$ is invariant.

Therefore, a computationally efficient form of the unique solution is
given by
\begin{align}
Y_* &= \argmax_{Y \in {\cal F}} f(Y;\;X) = \tilde{\ones}_n q_*^T + X Q_\perp^* (Q_\perp^*)^T \\
&= \tilde{\ones}_n q_*^T + X (I-q_*q_*^T),
\label{projection}
\end{align}
\begin{equation}
 q_* = \frac{X^T\tilde{\ones}_n}{\lVert X^T\tilde{\ones}_n \rVert_2}.
\end{equation}

\section{Numerical Experiments}
\subsection{Empirical Evaluation Techniques}
ARPPG was evaluated using a family of synthetic benchmark networks and real-world networks 
by comparing its performance to that of three state-of-the-art algorithms for community detection: the GN algorithm \cite{newman2004fast}, the Infomap algorithm \cite{rosvall2008maps} and the Louvain method \cite{blondel2008fast}.
The GN and Louvain methods were applied to maximizing the modularity $Q=\frac{1}{2m}\tr(X^TMX)$, where $X$ is an assignment matrix that specifies a partitioning of the nodes into communities. Even though the Infomap method was not designed to maximize the modularity, it is one of the best performing methods, see \cite{lancichinetti2011finding}. So, we also compared our algorithm with it. ARPPG maximized the cost function defined earlier based on modularity and a sparsity penalty term.


The assignments of nodes to communities produced by each algorithm for a given problem
were compared using their modularity values. However, since the modularity used
 here is one of many cost functions in the literature that heuristically define preferred assignments,
a metric independent of the cost function was used to assess the quality of the assignments.
A ground truth assignment of nodes to communities is associated with each benchmark graph. Given the ground truth, normalized mutual information (NMI) \cite{danon2005comparing} was used to compare the
quality of the communities.
NMI is a similarity measure between two partitions $X$ and $Y$ that represents their normalized mutual entropy and is defined
$$NMI(X, Y)=\frac{2 \mathcal{I}(X, Y)}{\mathcal{H}(X)+\mathcal{H}(Y)},$$
where $\mathcal{H}(X)$  is the entropy of the partition $X$ and $\mathcal{I}(X, Y)$ is the mutual information of the partitions $X$ and $Y$ given by 
\begin{gather*}
\mathcal{H}(X)=-\sum_u{\frac{n_u}{N}{\log\frac{n_u}{N}}},\\
\mathcal{I}(X, Y)=\sum_{u,v}{\frac{n_{uv}}{N}\log \left( N \frac{n_{uv}}{n_un_v}\right)},
\end{gather*}
with $n_u$ the number of nodes in community $u$ and $n_{uv}$ the number of common nodes in community $u$ of partition $X$ and community $v$ of partition $Y$. The value of NMI is in $[0,1]$ with larger values indicating higher similarity.

To correct the measures for randomness, it is necessary to specify a model according to which random partitions are generated. So, we used the adjusted mutual information (AMI) \cite{vinh2010information} as another measurement to assess the quality of the assignments of nodes to communities.  The AMI is defined to be 
$$AMI(X,Y)={\frac {\mathcal{I}(X,Y)-E\{\mathcal{I}(X,Y)\}}{\max {\{\mathcal{H}(X),\mathcal{H}(Y)\}}-E\{\mathcal{I}(X,Y)\}}},$$
where
\begin{align*}
E\{\mathcal{I}(X, &Y)\}=\sum_{u, v}\sum_{n_{uv}=(a_{u}+b_{v}-N)^{+}}^{\min(a_{u},b_{v})}{\frac {n_{uv}}{N}}\log \left({\frac {N\cdot n_{uv}}{a_{u}b_{v}}}\right)\times \\
&\frac {a_{u}!b_{v}!(N-a_{u})!(N-b_{v})!}{N!n_{{uv}}!(a_{u}-n_{{uv}})!(b_{v}-n_{{uv}})!(N-a_{u}-b_{v}+n_{{uv})!}}\\
\end{align*}
by adopting a hypergeometric model of randomness, 
where  $(a_{u}+b_{v}-N)^{+}$ denotes $\max(1,a_{u}+b_{v}-N)$, and $a_{u}=\sum _{v}n_{{uv}}$ and $ b_{v}=\sum _{u}n_{{uv}}$.

The synthetic benchmarks have clearly defined ground truth based on intracommunity connectivity
graphs  that are strongly connected but not necessarily completely connected as in our ideal
case defined above. The members of the family of networks are defined by a parameter that
makes the network have an increasingly ill-defined community structure. As a result, any reasonable
algorithm should detect community structure accurately when it is well-defined and the discrimination
ability of the algorithm is tested as the definition degrades. Additionally, we must consider the robustness of the combinatorial algorithms relative to their runtime choices, e.g., the particular random walks used in Infomap or the order and manner in which one-node moves are considered in the Louvain method. Similarly, ARPPG and other algorithms based on optimization over a continuous domain are dependent on their initial conditions or other strategies to avoid finding an unacceptable local maximum.

For a network representing real-world relationships there
can be ground truth based on a clear definition of the entities that define the
nodes, empirical observations such as observed social behavior, or 
classifications based on opinions of human observers who may or may not be experts in a discipline related to the information.  As a result, different cost functions may characterize the desired ground truth with different levels of accuracy. The use of the geometry, the sparsity constraint and continuation is an attempt to improve the robustness and aid in the selection of parameters such as the number of communities and the penalty parameter.

\subsection{Synthetic Benchmarks}



The generalized LFR benchmark graphs \cite{lancichinetti2009community} were used as the synthetic network benchmarks. These subsume the well-known benchmark proposed by Girvan and Newman \cite{girvan2002community} and are more challenging for community detection algorithms.
In the construction of the benchmark graphs, each node has a probability $p_{in}$ of being connected to nodes in its group and a probability $p_{out}$ of being connected to nodes in different groups. If $p_{in}>p_{out}$, the groups are communities, otherwise, the network is essentially a random graph without community structure. A power law distribution is used.

The condition $p_{in}>p_{out}$ can be translated into a condition on the mixing parameter $\mu$, which expresses the ratio between the external degree of a node with respect to its community and the total degree of the node \cite{lancichinetti2009community}:
$$\mu=\frac{k_i^{out}}{k_i^{in}+k_i^{out}}<\frac{N-n_c}{N},$$
where $k_i^{in}$ is the number of neighbors of node $i$ that belong to its community $c$ and $k_i^{out}$ the number of neighbors of $i$ that belong to the other communities, $N$ is the number of nodes, $n_c$ is the number of nodes of the community $c$.

Setting $\mu=0$, gives a graph defining a ground truth where the communities are strongly connected components and there are no edges between the communities. This is more challenging than the ideal ground truth of communities that are cliques used to motivate the optimization problem. For any value of $\mu >0$, the graph also has an associated ground truth but the mixing causes the community structure to be less clearly defined.
For the LFR benchmarks, the networks have $N=1000$ nodes, the average node degree is  $20$, the maximum node degree is $50$, the communities have between $20$ and $100$ nodes, the exponent of the degree power law distribution is $-2$, and the exponent of the community size power law distribution is $-1$.
The numbers of communities for the LFR benchmarks are around $20$.

\subsubsection{Results for the LFR networks:}
For the LFR benchmark with $\mu=0$, as expected and required, all four algorithms
have $NMI=1$, the same modularity value and the same assignment to $q_{true} = 24$ strongly 
connected communities.  ARPPG requires the desired number of communities as a parameter 
value and in this case it was taken as $q = q_{true} = 24$. The choice of an initial $q$ and
the development of a dynamic adaptation strategy are key ongoing tasks for ARPPG.
There is promising evidence that it is possible. For $\mu=0$ and ARPPG run with 
$q=25$ and $q=26$, i.e., near $q_{true}$, the modularity decreases as $q$ increases.
The final values of NMI for $q=25$ and $q=26$ change only slightly $0.99$ and $0.98$ respectively.
Of course this information is not available for the algorithm to use, but it is due to the fact that
the partitioning for $q=25$ and $q=26$  are nested in the partitioning for $q = q_{true} = 24$,
i.e., the extra communities are refinements of the $24$ by splitting without crossing
the ideal community boundaries.
Any nodes that are not in the same community in the ideal partitioning remain in different
communities in the refined partitions. This information can be detected by the algorithm
and used to guide adjustment of $q$ while revealing a hierarchical structure relevant to 
discussion of resolution limits \cite{fortunato2007resolution} and alternative cost functions, e.g., the constant
Potts model \cite{traag2011narrow}.

The algorithms were also tested with multiple nonzero values of $\mu$.
The values of 
NMI and modularity are shown in
Table \ref{tb:performance on LFR}
where ARPPG uses $q=q_{true}$ determined by
the network for each value of $\mu$.
All four methods determine the ground truth community assignments for the networks
with $\mu \leq 0.3$. For $\mu=0.4$ and $\mu=0.5$ three methods determine the associated ground truths
and one comes very close: GN with $NMI=0.99$, $AMI=0.99$ and ARPPG with $NMI=0.99$, $AMI=0.99$ respectively.

ARPPG using $q \neq q_{true}$ for $\mu \leq 0.4$ demonstrates trends like those for
$\mu=0$ upon which a $q$ adaptation strategy might be built.
As $q$ increases from $q_{true}$, NMI, AMI and modularity decrease at a rate that
increases as $\mu$ increases. The partitions are nested, then only slightly not nested with
one or two nodes crossing communities of the ground truth assignment, and finally
with a significant loss of nesting.

For the noisy cases in Table \ref{tb:performance on LFR}, GN degrades quickly while ARPPG and the Louvain method degrade more slowly.
Infomap achieves an $NMI=1$, $AMI=1$ until $\mu=0.5$ then drops to near $0$.
The performance of Infomap and the Louvain method are sensitive to their runtime decisions,
e.g., the Infomap performance here uses the heuristic available in the publicly distributed
code of running the method multiple times and choosing the ``best'' result. ARPPG, on the
other hand, with its continuation strategy and initial condition selection using RNewton
was seen to be remarkably robust even in the noisy situations.

 \begin{table}[htbp]
\begin{center}
\caption{Performance on LFR Bechmark Networks}\label{tb:performance on LFR}
\resizebox{\textwidth}{!}
{
\begin{tabular}{l |c |c| c |c |c |c| c| c|c}
\hline
\hline
\multirow{2}{*}{Methods} & \multirow{2}{*}{Measurements} & \multicolumn{8}{c}{The mixing parameter $\mu$}\\
\cline{3-10}
&  & 0.1 & 0.2 & 0.3 & 0.4 & 0.5 & 0.6 & 0.7 &  0.8\\
    \hline
\multirow{3}{*}{GN } & NMI & 1 & 1 & 1 & 0.9972 & 0.8694 & 0.6679 & 0.4932 & 0.4886\\
 & AMI & 1 & 1  & 1 & 0.9962 & 0.7202 & 0.2539 & 0.0142 & 0.0031\\
 & Modularity & 0.8254 & 0.7268 &  0.6283 & 0.5280 & 0.3579 &0.1230& 0.0393 &0.0329\\ \hline
 
\multirow{3}{*}{Infomap} & NMI & 1 & 1 & 1 & 1 & 1 & 0 & 0 & 0\\ 
 & AMI & 1 & 1  & 1 & 1 & 1 & 0 & 0 & 0 \\
  & Modularity & 0.8254 & 0.7268 & 0.6283 & 0.5288 & 0.4440 & 0 & 0 & 0\\ \hline
 
\multirow{3}{*}{Louvain}  & NMI & 1 & 1 & 1 & 1 & 1 & 0.9527 & 0.2192 & 0.0677\\
 & AMI & 1 & 1  & 1 & 1 & 1 & 0.9107 & 0.1748 & 0.0267 \\
 & Modularity & 0.8254 & 0.7268 & 0.6283 & 0.5288 & 0.4440 & 0.3390 & 0.2093 & 0.1921\\ \hline
 
 \multirow{3}{*}{ARPPG}  & NMI & 1 & 1 & 1 & 1 & 0.9935 & 0.8811 & 0.3422 & 0.0967\\
 & AMI & 1 & 1  & 1 & 1 & 0.9927 & 0.8651 & 0.3014 & 0.0473 \\
 & Modularity & 0.8254 & 0.7268 & 0.6283 & 0.5288 & 0.4427 & 0.3239 & 0.1712 & 0.1355\\ \hline
 \hline
\end{tabular}}
\end{center}
\end{table}

\subsection{Real World Networks}
Three widely used real-world networks were used to assess the performance of ARPPG.
The first is an American college football network \cite{girvan2002community}, in which the nodes represent football teams, and an edge exists between the nodes if there is a match between two teams. The ground truth community assignment is given by the membership in the same athletic conference, i.e.,
indisputable observations.
The second is Zachary's karate club network \cite{zachary1977information}, which is an undirected social network of friendship between 34 members of a karate club at a university. Edges connect individuals who were observed to interact outside the activities of the karate club. The ground truth is based on the splitting of the membership into $2$ new disjoint karate clubs. However, there is a second ground truth based of $4$ communities of $2$ disjoint social groups within each of the $2$ new clubs.  The $2$ community ground truth is defined by indisputable observation, the $4$ community ground truth is based
on less precise social interaction data.
The third is the Polbooks network \cite{newman2006modularity} of books about US politics published around the time of the 2004 presidential election and sold by the online bookseller Amazon.com. Edges between books represent frequent co-purchasing of books by the same buyers. The ground truth is determined by the subjective classification of the books by a non-expert human observer.
Given the difference in the level of certainty becoming increasingly debatable in these benchmarks, it is expected that detecting communities should be more difficult and dependent on cost function selection and algorithm tuning for each of the three in turn.

 \begin{table}[htbp]
\begin{center}
\caption{Performance on Real-World Networks (the best performance is in bold), where $n$ is the number of nodes, $m$ is the number of edges, $q_{true}$ is the number of ground truth communities and numbers in parentheses are the numbers of communities detected. For ARPPG the numbers in parentheses are also the values used for ARPPG's parameter $q$.}\label{tb:performance on real}
\resizebox{\textwidth}{!}
{
\begin{tabular}{l |c |c| c |c |c |c| c| c| c|c|c }
\hline
\hline
Datasets & n & m & $q_{true}$ & Measurements & GN & Infomap & Louvain &  \multicolumn{4}{c}{ARPPG}\\
    \hline
\multirow{3}{*}{Football}  & \multirow{3}{*}{115} & \multirow{3}{*}{613} & \multirow{3}{*}{12} & NMI & 0.879(10) & \textbf{0.924}(12)  & 0.890(10) & \textbf{0.924}(12) & 0.911(13) & 0.912(14) & 0.882(10)\\
 & & & & AMI & 0.802(10) & \textbf{0.898}(12) & 0.821(10) & \textbf{0.898}(12) & 0.861(13) & 0.848(14) & 0.813(10)\\ 
 & & & & Modularity & 0.600(10) & 0.601(12) & \textbf{0.605}(10) & 0.601(12) & 0.581(13) & 0.566(14) & 0.596(10)\\ \hline

\multirow{3}{*}{Karate } &\multirow{3}{*}{ 34 }& \multirow{3}{*}{78} & \multirow{3}{*}{2} & NMI & 0.580(5) & 0.700(3)  & 0.587(4) & \textbf{1.000}(2) & 0.811(3) & 0.687(4) & 0.542(5)\\
& & & & AMI & 0.402(5) & 0.579(3) & 0.425(4) & \textbf{1.000}(2)  & 0.672(3) & 0.505(4) & 0.364(5)\\
 & & & & Modularity & 0.401(5) & 0.402(3) & 0.419(4) & 0.372(2)  & 0.373(3) & \textbf{0.420}(4) & 0.382(5)\\ \hline
 
\multirow{3}{*}{Polbooks} & \multirow{3}{*}{105} & \multirow{3}{*}{441} & \multirow{3}{*}{3} & NMI & 0.559(5) & 0.494(6)  & 0.537(5) & \textbf{0.565}(3) & 0.503(4) & 0.465(5) & 0.439(6)\\ 
& & & & AMI & 0.488(5) & 0.390(6) & 0.458(5) & \textbf{0.535}(3) & 0.424(4) & 0.362(5) & 0.323(6)\\
 & & & & Modularity & 0.517(5) & 0.523(6) & \textbf{0.527}(5) & {0.508}(3) & 0.504(4) & 0.510(5) & 0.505(6)\\ \hline
 \hline
\end{tabular}}
\end{center}
\end{table}

Table \ref{tb:performance on real} summarizes the performance on the real-world networks.
Note that overall modularity values for the community assignments produced are significantly
lower than those for the synthetic networks and the different assignments produced all have similar modularity values with significantly different quality as measured by NMI and AMI.
This is most pronounced for the opinion-based ground truth of the Polbooks network as expected.
For the football network, ARPPG using $q=q_{true}$ produces an assignment close to the ground truth. Infomap produces the same $12$ community assignment but requires multiple runs, some of which produce significantly different assignments. GN and the Louvain method do not get the correct number of communities despite achieving a value of modularity close to that from the other algorithms. ARPPG run with $q \neq q_{true}$ exhibits the same trends on modularity and nesting discussed for the synthetic networks as desired.

For the karate club network, only ARPPG with $q=q_{true}=2$ produces the ground truth with $2$ communities.  When ARPPG is run with $q\neq q_{true}$ it exhibits the desired nesting trends and, in particular,
for $q=4$ it produces the second ground truth known for the network. (The NMI and AMI in the table is not $1$ for that case because it is the $4$ community ground truth compared to the $2$ community ground truth.)
The Louvain algorithm produces different $4$ community assignments depending on the order of traversal of the nodes. The $4$ community ground truth is one of them but the one in the table are not quite the same as is seen from the NMI and AMI differing from that of ARPPG.
Infomap produces different community assignments with varying numbers of communities in different runs. The result in the table is the best one.
As expected, the Polbooks network is the most difficult. Modularity does not predict well the quality of the assignment measured by NMI and AMI. Even ARPPG with $q=q_{true}$ does not produce an assignment as close to ground truth as it does for the other two networks. The fact that modularity does not clearly indicate the ground truth is also seen in the trends for ARPPG with $q \neq q_{true}$. Nesting is not observed and the best modularity is observed for $q = 5 \neq q_{true} = 3$.

\section{Conclusion}
In this paper, we propose a new Riemannian projected proximal gradient method applied to modularity with a convex nonsmooth sparsity penalty term for community detection. Numerical results show that ARPPG is competitive with state-of-the-art algorithms in terms of quality of assignment and robustness. Observations of performance as algorithm parameters vary provide leading evidence that a parameter adaptation strategy and an efficient implementation are feasible. 

\section*{Acknowledgment}
This paper was partially supported by the U.S. National Science Foundation under grant DBI 1934157. The author Wen Huang was partially supported by the Fundamental Research Funds for the Central Universities (NO. 20720190060). Part of this work was performed while the 
 author Kyle A. Gallivan was a visiting professor at UC Louvain, funded by the Science and Technology Sector, with additional support by the Netherlands Organization for Scientific Research.

\bibliographystyle{plain}                                                     


\end{document}